\documentclass[11pt]{article}

\usepackage{caption}
\usepackage{graphicx}
\graphicspath{{Figures/}}
\usepackage{lscape}
\usepackage{rotating}
\usepackage[table]{xcolor}
 \usepackage{amsmath,amsfonts,amssymb,mathrsfs,amsthm}
\usepackage{bm}
\usepackage{natbib}
\usepackage{booktabs}
\usepackage{pdflscape}
\usepackage{todonotes}

\usepackage{authblk}

\newtheorem{theorem}{Theorem}

\theoremstyle{definition}

\newtheorem{assumption}{Assumption}

\newcommand{\E}{\mathbb{E}}
\newcommand{\R}{\mathbb{R}}

\newcommand{\N}{\mathbb{N}}

\newcommand{\Sigmanotcov}{\Omega}

\newcommand{\var}{\operatorname{var}}
\newcommand{\VaR}{\operatorname{VaR}}
\newcommand{\cov}{\operatorname{cov}}
\newcommand{\sd}{\operatorname{sd}}
\newcommand{\ES}{\operatorname{ES}}

\renewcommand{\leq}{\leqslant}
\renewcommand{\geq}{\geqslant}

\newcommand{\rd}{\mathrm{d}}

\newcommand{\alphanew}{\ensuremath{\theta}}

\begin{document}
\title{On the Basel Liquidity Formula for Elliptical Distributions}

\author{Janine Balter}
\affil{Deutsche Bundesbank\footnote{The opinions expressed
    in this paper are those of the author and do not necessarily
    reflect views shared by the Deutsche Bundesbank or its staff.}}

\author{Alexander J.\ McNeil}
\affil{The York Management School, University of York\footnote{Corresponding author: Alexander J.\ McNeil, The York Management School,
  University of York, Freboys
  Lane, York YO10 5GD, UK, +44 (0) 1904 325307, \texttt{alexander.mcneil@york.ac.uk}}}

\maketitle
\begin{abstract}
A justification of the Basel liquidity formula for risk capital in the
trading book is given under the assumption that market risk-factor
changes form a Gaussian white noise process over 10-day time steps
and changes to P\&L are linear in the risk-factor changes.  A generalization of the formula is derived
under the more general assumption that risk-factor changes are
multivariate elliptical. It is shown that the Basel formula tends
to be conservative when the elliptical distributions are from the heavier-tailed
generalized hyperbolic family. As a by-product of the analysis a Fourier approach to
calculating expected shortfall for general symmetric loss
distributions is developed.
\end{abstract}

\noindent {\it Keywords}\/: Basel accords; liquidity
risk; risk measures; expected shortfall; elliptical distributions;
generalized hyperbolic distributions

\clearpage

\section{Introduction}
As a result of the fundamental review of the trading book (FRTB) \citep{bib:basel-13} a new minimum capital standard for the trading
book has emerged \citep{bib:basel-16}. Under this standard banks are now required to calculate a
liquidity-adjusted expected shortfall risk measure on a daily
basis. This calculation is carried out at both the level of the whole
trading book and the level of individual desks using an aggregation formula that is
based on the concepts of liquidity horizons and square-root-of-time
scaling.

Every risk factor affecting the value of positions in the trading book
or desk is assigned to a unique liquidity horizon $\text{LH}_j$ which may be 10, 20, 40, 60
or 120 days. These horizons are conservative estimates of the amount
of time that would be required to execute trades that would eliminate
the portfolio's sensitivity to
changes in these risk
factors during a period of market
illiquidity. For example, risk factors for the equity price risk of
large-cap stocks are assigned to the shortest horizon of 10 days;
equity volatility risk factors for large-cap stocks are given a risk
horizon of 20 days; risk factors for structured credit instruments
have the longest liquidity horizon of 120 days.

The liquidity formula reflects the prevailing method
of risk calculation in the banking industry in which changes in P\&L
for trading book positions
are modelled in terms of sensitivities to risk factors. Expected
shortfall charges are calculated with respect to shocks to risk
factors with particular liquidity horizons while other risk factors are
held constant. To make the calculation explicit, we give the formula and notation as published on page 52 of the
revised capital standard~\citep{bib:basel-16}.
\begin{itemize}
\item let $T=\text{LH}_1$ denote the so-called base liquidity horizon
  of 10 days.
\item Let $\text{ES}_T(P)$ denote the expected shortfall at horizon
  $T$ for a portfolio $P$ with respect to shocks to all risk factors
  to which the positions in the portfolio are exposed.
\item Let $\text{ES}_T(P,j)$ denote the expected shortfall at horizon
  $T$ for a portfolio $P$ with respect to shocks to the risk factors
  which have a liquidity horizon of length $\text{LH}_j$ or greater,
  with all other risk factors held fixed.
\end{itemize}
The liquidity-adjusted expected shortfall is
\begin{equation}
  \label{eq:15}
  \text{ES} = \sqrt{\left(\ES_T(P)\right)^2 + \sum_{j\geq 2}\left( \ES_T(P,j)\sqrt{\frac{\text{LH}_j-\text{LH}_{j-1}}{T}}\right)^2}\,.
\end{equation}

The first objective of this paper is to provide a principles-based
derivation of this formula that relates it to the concept of expected shortfall as a risk measure applied to
a loss distribution or P\&L distribution. Most practitioners know that
an assumption of normality underlies the
formula. We make it precise that the formula can be justified by
assuming that risk-factor changes over time steps equal to the base
liquidity horizon form a multivariate Gaussian white
noise with mean zero and portfolio losses are all attributable to first-order
(delta) sensitivities to the risk-factor changes.

The second and major objective of the paper is to analyse the formula under the
more general assumption that risk-factor changes have a multivariate
elliptical distribution. This allows us to consider some particular
cases with heavy tails and
tail dependencies that might be considered more realistic models for
market risk-factor changes. 

Many results in QRM continue to hold when multivariate
normal assumptions are generalized to multivariate elliptical
assumptions. In particular, when losses are linear in a set of
underlying elliptically-distributed risk factors, aggregation of risk measures across
different business lines, desks or risk factors can generally be based on a common
formulaic approach, regardless of the exact choice of elliptical
distribution; see Chapter 8 of~\citet{bib:mcneil-frey-embrechts-15}. The
difference in the current paper is that aggregation takes place, not
only across risk factors, but also across time and therefore a `central limit
effect' takes place. We will show that~\eqref{eq:15} is
in fact a conservative aggregation rule for the popular generalized
hyperbolic family of heavier-tailed
elliptical assumptions and we will give a
generalization of the rule that holds for all elliptical distributions.

As a by-product of our analyses we also demonstrate a new 
approach to calculating VaR and expected shortfall for 
symmetric distributions with a known characteristic function. This
approach is particularly useful in cases where we take convolutions of elliptically
distributed random vectors and lose the ability to write simple
closed-form expressions for their probability densities.

We present all ideas in terms of the standard probabilistic approach
to risk measures. Losses (or P\&L variables) are represented by random
variables $L$. Expected shortfall ($\ES_\alpha$) and value-at-risk
($\VaR_\alpha$) at level $\alpha$ are risk measures
applied to $L$. If $F_L$ denotes the distribution function of $L$ and $q_\alpha$
the corresponding quantile function, they
are given by $\VaR_\alpha(L) = q_\alpha(F_L)$ and $\ES_\alpha(L) = \tfrac{1}{1-\alpha}
\int_\alpha^1 q_u(F_L) \rd u$. If $F_L$ is continuous then the formula
$\ES_\alpha(L) =\E(L \mid L  \geq \VaR_\alpha(L))$ also holds.

\section{Justifying and extending the Basel liquidity formula}

Let $(\bm{X}_t)$ be a $d$-dimensional time series of risk-factor
changes for all relevant risk factors and assume that these are all
defined in terms of simple differences or log-differences. We
interpret $\bm{X}_{t+1}$ as the vector of risk-factor changes over the
time step $[t,t+1]$. In practice this time step will be equal to the
base liquidity horizon of 10 days. 

For $h \in \N$, the risk factor changes over the
time step $[t,t+h]$ are given
additively by
\begin{equation}\label{eq:5}
  \bm{X}_{[t,t+h]} := \sum_{j=1}^{h}
\bm{X}_{t+j}.
\end{equation}
Without loss of generality let the risk calculation be made at time
$t=0$. We make the following assumptions.

\begin{assumption}\label{assumption:just-basel-liqu}
  \begin{enumerate}
  \item[(i)]
The risk-factor changes $(\bm{X}_t)$ form a stationary white noise
process (a serially uncorrelated process) with mean zero and covariance
matrix $\Sigma$.
\item[(ii)]
Each risk factor may be assigned to a unique liquidity bucket
$B_k$ defined by a liquidity horizon $h_k \in \N$, $k=1,\ldots,n$.
\item[(iii)]
In the event of a portfolio liquidation action the loss (or profit) attributable to risk factors in
bucket $B_k$ is given by 
$\bm{b}_k^\prime \bm{X}_{[0,h_k]}$ where $\bm{b}_k$ is a
weight vector with zeros in any position that corresponds to a risk
factor that is not in $B_k$.
 \end{enumerate}
\end{assumption}
\noindent Assumption~\ref{assumption:just-basel-liqu}(iii) contains the
linearity assumption and adopts the pessimistic view that
the full liquidity horizon $h_k$ is required to remove the portfolio's
sensitivity to all the risk factors in liquidity bucket $B_k$.

Under these assumptions we compute the portfolio loss $L$ over the
maximum time horizon $h_n$, which is the time required to remove the
portfolio's sensitivity to all risk factors. It follows from
Assumption~\ref{assumption:just-basel-liqu}(ii) and (iii) that
\begin{align}
  L &= \sum_{k=1}^n \bm{b}_k^\prime \bm{X}_{[0,h_k]} = 
 \sum_{k=1}^n \sum_{j=1}^{k}
\bm{b}_k^\prime \bm{X}_{[h_{j-1},h_j]}  
= \sum_{k=1}^n \sum_{j=k}^{n}
\bm{b}_j^\prime \bm{X}_{[h_{k-1},h_k]} 
=
\sum_{k=1}^n
\bm{\beta}_k ^\prime \bm{X}_{[h_{k-1},h_k]} 
\label{eq:L}
\end{align}
where $\bm{\beta}_k = \sum_{j=k}^n \bm{b}_j$ and $h_0=0$. The vector
$\bm{\beta}_k$ contains the weights for all risk factors in the union
of liquidity 
buckets $B_k\cup\cdots \cup B_n$.

Let us write $L_k :=
\bm{\beta}_k ^\prime \bm{X}_{[h_{k-1},h_k]}$ for $k=1,\ldots,n$ for
the
summands in the final expression in~\eqref{eq:L}. These
are uncorrelated by Assumption~\ref{assumption:just-basel-liqu}(i) and we
may easily calculate
that
\begin{equation}\label{eq:10}
  \var(L) = \sum_{k=1}^n \var(L_k) = \sum_{k=1}^n
\bm{\beta}_k^\prime \cov(\bm{X}_{[h_{k-1},h_k]}) \bm{\beta}_k
= \sum_{k=1}^n (h_k - h_{k-1})
\bm{\beta}_k^\prime \Sigma \bm{\beta}_k .
\end{equation}
 where the final step follows because~\eqref{eq:5} implies that
 $ \bm{X}_{[h_{k-1},h_k]}  = \sum_{j=1}^{h_k-h_{k-1}} \bm{X}_{h_{k-1}+j}$.





We now introduce random variables
\begin{equation}\label{eq:20}
  L^{(k)} = \bm{\beta}_k^\prime \bm{X}_{[0,h_1]} 
\end{equation}
for $k=1,\ldots,n$. These represent losses attributable to all risk factors in the union
of liquidity 
buckets $B_k\cup\cdots \cup B_n$ over the liquidity horizon
$h_1$. Note that the $L_k$ and $L^{(k)}$ variables differ (unless $k=1$). 
Since $\var(L^{(k)}) = h_1 \bm{\beta}_k^\prime \Sigma \bm{\beta}_k$,
we obtain from~\eqref{eq:10} the formula
\begin{equation}\label{eq:1}
  \sd(L) = 
\sqrt{ \sum_{k=1}^n \left(\sqrt{\frac{h_k - h_{k-1}}{h_1}}\sd(L^{(k)})\right)^2}.
\end{equation}

It may be noted that the presence of positive correlation between the
variables $L_k$ in~\eqref{eq:10}, caused by serial correlation in the
underlying risk-factor changes $\bm{X}_{[h_{k-1},h_k]}$, would tend to
lead to
the left-hand side of~\eqref{eq:1} being larger than the right-hand
side. Negative correlation would lead to it being smaller.



\subsection{The Gaussian case}\label{sec:gaussian-case}
Suppose that $(\bm{X}_t)$ is a Gaussian process; in this case
$(\bm{X}_t)$ is actually a strict white noise (a process of
independent and identically distributed vectors). It follows that $L_k \sim N(0,
  (h_k-h_{k-1})\beta_k^\prime \Sigma \beta_k)$ and the $L_k$ are
  independent for all $k$. Thus, by the convolution property for
  independent normals,
\begin{equation}\label{eq:3}
L \sim N\left(0, \sum_{k=1}^n
  (h_k-h_{k-1})\bm{\beta}_k^\prime \Sigma \bm{\beta}_k\right).
\end{equation}
Moreover, we clearly have $L^{(k)} \sim N(0,h_1\bm{\beta}_k^\prime \Sigma
\bm{\beta}_k)$. 

For any mean-zero normal random variable $V$ it is easy to show that
$\ES_\alpha(V) = c_\alpha \sd(V)$ where $c_\alpha =
\phi(\Phi^{-1}(\alpha))/(1-\alpha)$, $\phi$ denotes the density of the
standard normal distribution and $\Phi^{-1}(\alpha)$ denotes the
$\alpha$-quantile of the standard normal distribution function
$\Phi$~\citep[see][Chapter 2]{bib:mcneil-frey-embrechts-15}. 
It follows from~\eqref{eq:1} that
\begin{equation}\label{eq:4}
  \ES_\alpha(L) = 
\sqrt{ \sum_{k=1}^n \left(\sqrt{\frac{h_k - h_{k-1}}{h_1}}\ES_\alpha(L^{(k)})\right)^2}
\end{equation}
which is the proposed standard formula for the trading
book~\eqref{eq:15} rewritten in our notation.

\subsection{An extension to the formula for elliptical distributions}

In this section we assume a centred elliptical distribution for the
risk-factor changes, which subsumes the multivariate normal
distribution as a special case. In addition to
Assumption~\ref{assumption:just-basel-liqu} we assume that the
following holds.
\begin{assumption}\label{assumption:liqu-form-ellipt}
  \begin{enumerate}
  \item[(i)] The process $(\bm{X}_t)$  is a multivariate strict white
  noise (an iid process).
\item[(ii)] The distribution of $\bm{X}_t$ is elliptical with
location $\bm{0}$, positive-definite dispersion matrix $\Sigmanotcov$ and characteristic generator function
$\psi=\psi(s)$, written $\bm{X}_t \sim
E_d(\bm{0},\Sigmanotcov,\psi)$.
  \end{enumerate}
\end{assumption}

Assumption~\ref{assumption:liqu-form-ellipt}(i) may seem strong but 
in practice we assume that $(\bm{X}_t)$ is a process of
10-day returns so that the iid assumption, while unlikely to be true,
is less problematic than for daily financial returns. The assumption
is required in order to analyse convolutions of
elliptically distributed random vectors with different characteristic generators.

Assumption~\ref{assumption:liqu-form-ellipt}(ii) means that $\bm{X}_t = A\bm{Y}_t$ for some
matrix $A \in \R^{d \times d}$ satisfying $\Sigmanotcov = AA^\prime$ and
some random variable $\bm{Y}_t$
with characteristic function given by $\phi(\bm{s}) =
E(e^{i\bm{s}^\prime \bm{Y}_t}) =
\psi(\bm{s}^\prime\bm{s})$ for a function of a scalar variable $\psi$. $\bm{Y}_t$ is said to have a spherically
symmetric distribution, which is written $\bm{Y}_t \sim S_d(\psi)$. It is
important to note that $\Sigmanotcov$ is not the covariance
matrix of $\bm{X}_t$ unless the covariance matrix of $\bm{Y}_t$ is the
identity matrix; in general we have $\Sigma =  \var(Y)\Sigmanotcov$
where $Y \sim S_1(\psi)$. The class of elliptical distributions
contains a number of particular distributions which are popular models
for financial returns including the multivariate Student t and the
symmetric generalized hyperbolic distributions.
See~\citet{bib:fang-kotz-ng-90}
and~\citet{bib:mcneil-frey-embrechts-15} for further details of these
distributions. 

We need three key properties of an elliptical distribution for our
calculation. Let $\bm{X} \sim
E_d(\bm{0},\Sigmanotcov,\psi)$ and $\tilde{\bm{X}} \sim
E_d(\bm{0},\Sigmanotcov,\tilde{\psi})$ be independent elliptically-distributed
variables with the same dispersion matrix $\Sigmanotcov$ and possibly
different characteristic generators $\psi$ and $\tilde{\psi}$.
\begin{align}
&\bm{\beta}^\prime \bm{X}  \sim
E_1(0,\bm{\beta}^\prime\Sigmanotcov\bm{\beta},\psi) \quad \text{for $\bm{\beta}\in
\R^d$ and $\bm{\beta}\ne\bm{0}$.} \label{eq:6}\\
&\bm{X}  \sim
E_d(\bm{0},c\Sigmanotcov,\psi(s/c)) \quad \text{for any $c>0$.}\label{eq:7}\\
&\bm{X}+\tilde{\bm{X}} \sim E_d(\bm{0},\Sigmanotcov,\psi^*) \quad \text{where
    $\psi^*(s) = \psi(s)\tilde{\psi}(s)$.}\label{eq:8}
\end{align}
We will use~\eqref{eq:6} and~\eqref{eq:8} to find the characteristic
functions of elliptical ramdom vectors under linear combinations and
convolutions respectively. The property in~\eqref{eq:7} shows that we
have some discretion in how we represent the characteristic generator
of an elliptical random
variable in terms of its characteristic generator and its scaling. 

\begin{theorem}
Under Assumptions~\ref{assumption:just-basel-liqu}
and~\ref{assumption:liqu-form-ellipt} the loss $L$ in~\eqref{eq:L}
is a univariate spherical random variable $L \sim S_1(\psi_L)$ with
characteristic generator
\begin{equation}\label{eq:2}
\psi_L(s) = \prod_{k=1}^n \psi_k (s \bm{\beta}_k^\prime \Sigmanotcov \bm{\beta}_k),
\end{equation}
where
$\psi_k =
\psi^{h_k-h_{k-1}}$ for $k=1,\ldots,n$.

For $\alpha > 0.5$ the expected shortfall of $L$ is
related to the expected shortfall of
the variables $L^{(k)}$ in~\eqref{eq:20} by
\begin{equation}\label{eq:16}
  \ES_\alpha(L) =  \frac{c_{\alpha,\psi_L}}{c_{\alpha,\psi_1}}
\sqrt{ \sum_{k=1}^n \left(\sqrt{\frac{h_k - h_{k-1}}{h_1}}\ES_\alpha(L^{(k)})\right)^2}.
\end{equation}
where $c_{\alpha,\psi_L}$ represents the ratio of expected shortfall to standard
deviation for $L$ and $c_{\alpha,\psi_1}$ is the equivalent ratio for a
univariate spherical variable $Z \sim S_1(\psi_1)$.
\end{theorem}
\begin{proof}
We
need to derive the distributions of
\begin{equation}
  \label{eq:9}
  L_k =
\bm{\beta}_k ^\prime \bm{X}_{[h_{k-1},h_k]}, \quad L = \sum_{k=1}^n
L_k \quad\text{and}\quad  L^{(k)} = \bm{\beta}_k^\prime \bm{X}_{[0,h_1]}.
\end{equation}
First note that if $\bm{X}_t \sim
E_d(\bm{0},\Sigmanotcov,\psi)$ then it follows from~\eqref{eq:5}
and~\eqref{eq:8} that $\bm{X}_{[h_{k-1},h_k]} \sim
E_d(\bm{0},\Sigmanotcov,\psi_k)$ where $\psi_k = \psi^{h_k-h_{k-1}}$. 
Using~\eqref{eq:6} we have that
\begin{displaymath}
  L_k \sim E_1\left(0, 
\bm{\beta}_k^\prime \Sigmanotcov \bm{\beta}_k, \psi_k \right) \quad
\text{and}\quad L^{(k)} \sim
E_1(0,\beta_k^\prime\Sigmanotcov\beta_k,\psi_1).
\end{displaymath}
Using~\eqref{eq:7} we write the former as $L_k \sim E_1\left(0, 
1, \psi_k (s \bm{\beta}_k^\prime \Sigmanotcov \bm{\beta}_k)\right)$
or   $L_k \sim S_1\left(\psi_k (s \bm{\beta}_k^\prime \Sigmanotcov
  \bm{\beta}_k)\right)$ and then
use the convolution property~\eqref{eq:8} to conclude that $L \sim
S_1(\psi_L)$ where $\psi_L$ is given in~\eqref{eq:2}.

Now $\ES_\alpha(L^{(k)}) = \sqrt{\bm{\beta}_k^\prime \Sigmanotcov
  \bm{\beta}_k}\ES_\alpha(Z)$ and $\sd(L^{(k)}) = \sqrt{\bm{\beta}_k^\prime \Sigmanotcov
  \bm{\beta}_k}\sd(Z)$ where $Z \sim S_1(\psi_1)$. Hence it follows that $\ES_\alpha(L^{(k)}) =
c_{\alpha,\psi_1}\sd(L^{(k)})$ for all $k$ and
\begin{align*}
  \ES_\alpha(L) = c_{\alpha,\psi_L} \sd(L) 
&=  c_{\alpha,\psi_L} \sqrt{ \sum_{k=1}^n \left(\sqrt{\frac{h_k -
  h_{k-1}}{h_1}}\sd(L^{(k)})\right)^2} \\
&=  c_{\alpha,\psi_L} \sqrt{ \sum_{k=1}^n \left(\sqrt{\frac{h_k -
  h_{k-1}}{h_1}}\frac{\ES(L^{(k)})}{c_{\alpha,\psi_1}}\right)^2} 
\end{align*}
which yields~\eqref{eq:16}.
\end{proof}

It may be easily verified that when $\psi(s) = \exp(-s/2)$ (the Gaussian
case), the characteristic function $\phi(s) = \psi_L(s^2)$ implied by~\eqref{eq:2} is the
characteristic function of the normal distribution in~\eqref{eq:3}. In
this case the constants $c_{\alpha,\psi_L}$ and $c_{\alpha,\psi_1}$
are identical.


When the risk factors have a heavier-tailed distribution than normal
we expect that
$c_{\alpha,\psi_L} \leq c_{\alpha,\psi_1}$, due to the central limit
effect, so the Basel
liquidity formula should give an upper bound.

\section{Calculating the scaling ratio in practice}

We turn to the problem of calculating the ratio
$r_\alpha:= c_{\alpha,\psi_L}/c_{\alpha,\psi_1}$ when the underlying risk factors
have an elliptical distribution with generator $\psi$. 
To compute $c_{\alpha,\psi_1}$ we calculate the ratio
$\ES_\alpha(Z)/\sd(Z)$
for a univariate spherical random variable $Z$ with characteristic
generator $\psi_1 = \psi^{h_1}$.
To compute $c_{\alpha,\psi_L}$ we calculate the ratio
$\ES_\alpha(L)/\sd(L)$ for a univariate spherical variable $L$ with
characteristic generator given by~\eqref{eq:2}.

In general we will not be able to
calculate $\ES_\alpha(Z)$ and $\ES_\alpha(L)$ from the probability
densities of $Z$ and $L$, since these typically do not have simple closed
forms for the distributions of interest. In the following section we give results that can be
used to compute 
expected shortfall directly from the characteristic function of a spherical
random variable.

\subsection{Calculating expected shortfall by Fourier inversion}
A univariate spherical random variable $Y \sim S_1(\psi)$ is symmetric about the origin
with a real-valued even characteristic function given by $\phi_Y(s) :=
\psi(s^2)$. We give a general result that applies to univariate random variables
that are symmetric about the origin.

\begin{theorem}\label{theorem:fourier-inversion}
Let $Y$ be symmetrically distributed about the origin with an integrable characteristic function
$\phi_Y(s)$. Let $-\infty < a < b<\infty$. Then the following formulas hold:
\begin{align}
  f_Y(y) 
&= \frac{1}{\pi} \int_0^\infty \cos(sy) \phi_Y(s) {\rm d}
  s,\label{eq:density} \\
F_Y(y) &= \frac{1}{2} + \frac{1}{\pi} \int_0^\infty
         \frac{\sin(sy)}{s} \phi_Y(s){\rm d} s,\label{eq:df}\\
E(Y I_{\{a \leq Y\leq b\}}) &=\frac{1}{\pi} \int_0^\infty \frac{bs\sin(bs)+\cos(bs)-as\sin(as)-\cos(as)}{s^2}
  \phi_Y(s){\rm d} s\,.\label{eq:shortfall}
\end{align}
\end{theorem}
\begin{proof}
The
characteristic function $\phi_Y(s)$ of a random variables that is
symmetric about the origin is real-valued and even. If $\phi_Y$ is integrable then the density exists and the
standard Fourier inversion formula for the characteristic formula yields
\begin{displaymath}
  f_Y(y) = \frac{1}{2\pi} \int_{-\infty}^\infty e^{-isy} \phi_Y(s)
  {\rm d} s = \frac{1}{\pi} \int_0^\infty \cos(sy) \phi_Y(s) {\rm d} s.
\end{displaymath}
The formula~\eqref{eq:df} for the distribution function is obtained
from a
well-known representation of the distribution by~\citet{bib:gil-pelaez-51}.
To derive~\eqref{eq:shortfall} we observe that
\begin{align*}
  \int_a^b y f_Y(y) &= \frac{1}{\pi}\int_a^b \int_0^\infty y \cos(sy)
  \phi_Y(s){\rm d}s {\rm d}y \\
&= \frac{1}{\pi} \int_0^\infty\left( \int_a^b y\cos(sy)
 {\rm d}y  \right)\phi_Y(s){\rm d}s 
\end{align*}
by Fubini's Theorem
since $|y\cos(sy)\phi_Y(s)|\leq |y||\phi_Y(s)|$ and the latter is
integrable on $[a,b]\times[0,\infty)$.
The inner integral can be solved by parts to obtain
\begin{displaymath}
  \int_a^b y\cos(sy)
 {\rm d}y  =  \frac{bs\sin(bs)+\cos(bs)-as\sin(as)-\cos(as)}{s^2}
\end{displaymath}
and~\eqref{eq:shortfall} follows.
\end{proof}
These formulas permit the accurate evaluation of $\VaR_\alpha(Y)$ and expected
shortfall using one-dimensional integration. Calculation of
$\VaR_\alpha(Y)$ for $\alpha > 0.5$ is accomplished by numerical root finding
using~\eqref{eq:df}. If $\E|Y| < \infty$ for the
distribution in question, then expected shortfall is defined and it can be calculated by setting
$a=\VaR_\alpha(Y)$ and computing the limit
\begin{equation}\label{eq:12}
  \ES_\alpha(Y) = \lim_{b\to\infty}\frac{1}{\pi(1-\alpha)} \int_0^\infty \frac{bs\sin(bs)+\cos(bs)-as\sin(as)-\cos(as)}{s^2} \phi_Y(s){\rm d} s\,.
\end{equation}
Our experiments confirm that calculating the integral in~\eqref{eq:12}
for increasing $b$ does result in stable limiting values for
$\ES_\alpha(Y)$ which agree to a high level of accuracy with
theoretical values for well-known
distributions such as Student t.

\subsection{The case of generalized hyperbolic
  distributions}\label{sec:examples}

We will apply Theorem~\ref{theorem:fourier-inversion} to the family of
symmetric generalized hyperbolic (GH) distributions. This is a very
popular family for modelling financial returns and there are many useful
sources for the properties of these distributions 
including~\cite{bib:barndorff-nielsen-78},~\citet{bib:barndorff-nielsen-blaesild-81},~\citet{bib:eberlein-10}
and~\citet{bib:mcneil-frey-embrechts-15}. 

Let $\bm{Y}=(Y_1,\ldots,Y_d)^\prime$ have the
stochastic representation $\bm{Y}= \sqrt{W}\bm{V}$ where
$\bm{V}=(V_1,\ldots,V_d)^\prime$ is a vector of independent standard normal
variables and $W$ is an independent positive random variable with a
so-called generalized inverse Gaussian (GIG) distribution $W \sim N^{-}(\lambda,\chi,\kappa)$; see
formula~\eqref{eq:11} in the Appendix for the density of this distribution. The vector $\bm{Y}$ has a spherical distribution $\bm{Y} \sim
S_d(\psi)$, and any component $Y$ has a univariate spherical
distribution $Y \sim S_1(\psi)$, for a characteristic generator $\psi$ that depends on the particular
choice of the parameters $\lambda$, $\chi$ and $\kappa$. An elliptical
model of the kind described in
Assumption~\ref{assumption:liqu-form-ellipt}(ii) is obtained by taking
$\bm{X} = A \bm{Y}$ for $A \in \R^{d \times d}$ and satisfies $\bm{X} \sim E_d(\bm{0},\Omega,
\psi)$ where $\Omega = A A^\prime$. $\bm{X}$ is said to have a
$d$-dimensional symmetric generalized hyperbolic (GH) distribution.

To carry out our calculations it suffices to consider the single
component $Y$.
The variance of $Y$ satisfies $\var(Y) =
\E(W)$ and an explicit formula for the case where $\chi>0$ and $\kappa>0$
is given in~\eqref{eq:18}.  A formula for the characteristic
function $\phi_Y$ is given in~\eqref{eq:14} and the characteristic
generator of the elliptical family can be inferred from the identity $\psi(s^2) = \phi_Y(s)$.

We consider four special one-parameter cases of this distribution
resulting from particular choices of the parameters $\lambda$, $\chi$
and $\kappa$ of the GIG distribution:
\begin{enumerate}
\item The student t distribution with degree of freedom $\nu$. This corresponds to the case where
  $\kappa=0$, $\lambda = -\nu/2$ and $\chi=\nu$ or where $W$ has an
  inverse gamma distribution $W \sim \text{IG}(\nu/2,\nu/2)$. In this case
$\var(Y) = \nu/(\nu-2)$, provided $\nu>2$, and the characteristic
  function is given by~\eqref{eq:13} in the Appendix.
\item The variance gamma (VG) distribution. This corresponds to the case where
  $\chi=0$ or where $W$ has a
  gamma distribution $W \sim \text{Ga}(\lambda,\kappa/2)$. Without
  loss of generality we set the scaling parameter $\kappa=2$ so that
$\var(Y)=\lambda$.
The corresponding characteristic function is given by~\eqref{eq:167}.
\item The normal-inverse-Gaussian (NIG) distribution. This corresponds to the case where
  $\lambda = -1/2$. The distribution can be reparameterized in terms
  of $\alphanew = \sqrt{\chi\kappa}$ and $\chi$; the
  latter parameter can be treated as a scaling parameter and set to
  one. The variance is then $\var(Y) = \alphanew^{-1}$ and the
  characteristic function is given by~\eqref{eq:17}.
\item The hyperbolic (Hyp) distribution. This corresponds to the case where
  $\lambda = 1$. The distribution can be reparameterized in exactly the same
  way as the NIG distribution. The variance is
  $\var(Y)=\alphanew^{-1}K_2(\alphanew)/K_1(\alphanew)$ and the 
characteristic function is given by~\eqref{eq:177}.
 
\end{enumerate}

\subsection{Summary of the steps in the calculation}
We return to the problem of calculating the scaling ratios
$r_\alpha = c_{\alpha,\psi_L}/c_{\alpha,\psi_1}$ in~\eqref{eq:16} when the
underlying risk-factor returns have symmetric distributions in the
multivariate generalized hyperbolic family.

We recall the basic components that are required for the calculation: $Y \sim S_1(\psi)$ is spherically distributed with known standard
  deviation $\sd(Y)$ and known characteristic
  function $\phi_Y(s) = \psi(s^2)$; $Z\sim S_1(\psi_1)$ where
  $\psi_1 = \psi^{h_1}$; $L \sim
  S_1(\psi_L)$ where $\psi_L$ is given in~\eqref{eq:2}.
 The steps are:
\begin{enumerate}
\item Calculate $\ES_\alpha(Z)$ using~\eqref{eq:12} and $\phi_Z(s) = \phi_Y^{h_1}(s)$. 
\item Calculate $\sd(Z)=\sqrt{h_1}\sd(Y)$. 
\item Hence calculate $
c_{\alpha,\psi_1} = \ES_\alpha(Z)/\sd(Z)$.
\item Calculate $\ES_\alpha(L)$ using~\eqref{eq:12} and the fact that
  \begin{displaymath}
    \phi_L(s) = \prod_{k=1}^n
    \phi_Y^{h_k-h_{k-1}}\left(s\sqrt{\bm{\beta}_k^\prime \Sigmanotcov\bm{\beta}_k}\right).
  \end{displaymath}
\item Calculate $\sd(L)$ using the formula
  \begin{displaymath}
    \sd(L) = \sd(Y)\sqrt{  \sum_{k=1}^n (h_k-h_{k-1}) \bm{\beta}_k^\prime \Sigmanotcov\bm{\beta}_k}.
  \end{displaymath}
\item Hence calculate $ c_{\alpha,\psi_L}= \ES_\alpha(L)/\sd(L)$.
\item Hence calculate the ratio $r_\alpha = c_{\alpha,\psi_L}/c_{\alpha,\psi_1}$.
\end{enumerate}

\section{Results}

\subsection{Design of experiments}

In order to calibrate our model distributions, we use 2132
observations of adjusted daily closing prices for the S\&P500 index,
from 17.7.2007 to 31.12.2015, which have been converted to two-weekly
log-returns (conforming approximately to 10 trading days, the base liquidity horizon required under FRTB).

We fit the various distributions discussed in
Section~\ref{sec:examples}  to the 10-day return data using the R package \texttt{ghyp}.
Table \ref{tab:param} gives the estimated shape
parameters for the distributions of interest; scale parameters are
not required in our analysis. Note that we also
confirm that the calculations for the Gaussian case yield a ratio of 1,
as a check on our implementation.

\begin{table}[htbp]
  \centering
  \begin{tabular}{*{1}{l}*{2}{r}*{1}{l}}
    \toprule
    Distribution \textbar\ Parameters & \multicolumn{1}{c}{$\lambda$} &
                                                                 \multicolumn{1}{c}{$\alphanew$} & Remarks\\
    \midrule
    t & -1.46 &  & $\nu=-2\lambda$\\
    NIG & -0.5 & 0.49 & $\lambda$ fixed\\
    Hyp & 1 & 0.11 & $\lambda$ fixed \\
    VG & 0.95 & & $\kappa=2$ \\
    \bottomrule
  \end{tabular}
\caption{Distribution parameters used in the calculation
  experiments. These have been derived by fitting these distributions
  to two-weekly log-returns of the S\&P500 index over the period from 17.07.2007 to 31.12.2015.}
\label{tab:param}  
\end{table}

We carry out two experiments: 
\begin{itemize} 
\item In the first, we consider two risk factors, one in $B_1$ with  a
  liquidity horizon of 10 days ($h_1=1$) and the other in $B_2$ with a
  liquidity horizon of 20 days ($h_2=2$). The dispersion matrix
  $\Sigmanotcov$ is either taken to be the identity $\Sigmanotcov =
  I_2$ (no correlation) or a correlation matrix with correlation $\rho=0.5$.  
\item The second experiment follows in the same fashion but we assume
  there are 5 risk factors with liquidity horizons 10, 20, 40, 60 and 120 days
  ($h_1=1, h_2=2, h_3=4, h_4=6, h_5=12$). We consider both the case
  where $\Sigmanotcov =
  I_5$ and the case where $\Sigmanotcov$ is an equicorrelation matrix
  with element $\rho=0.5$.
\end{itemize} 

We present values of $c_{\alpha,\psi_{1}}$, $c_{\alpha,\psi_L}$ as
well as the scaling ratio $r_\alpha$ for various confidence levels $\alpha$. 
The case of two risk factors is reported in Table~\ref{tab: exp1
  bck=2} and the case of five risk factors is reported in Table~\ref{tab: exp1 bck=5}.

\subsection{Results}

In both tables it is clear that the scaling ratios are less than 
one for
all non-Gaussian cases meaning that the Basel liquidity formula is
indeed conservative when the risk factors have a multivariate
elliptical distribution from one of the four generalized hyperbolic
sub-families considered in Section~\ref{sec:examples} and
Table~\ref{tab:param}.

The second experiment with five liquidity buckets leads in general to
smaller values for the scaling ratios than the first experiment with
two buckets. Thus the degree of conservatism of the formula
increases with the number of liquidity buckets. This is in line with
the increase in the central limit effect as we aggregate over more
time periods.

Introducing correlation leads to an increase in the constants
$c_{\alpha,\psi_L}$ and hence an increase in the scaling ratio. In
other words,  the weaker the correlation, the more conservative the
liquidity formula. To understand why this is the case, note that the constants $c_{\alpha,\psi_L}$ 
depend on the
characteristic generator $\psi_L$ in (\ref{eq:2}) and hence on the
set of values $\{\bm{\beta}_k^\prime \Sigmanotcov
\bm{\beta}_k,\;k=1,\ldots,n\}$. By considering formula~\eqref{eq:10} we
can think of these as
the relative weights attached to each of the $n$ liquidity buckets.  When $\rho=0$ these weights are $(5, 4, 3, 2, 1)$ but
when $\rho=0.5$ they are $(15, 10, 6, 3, 1)$. 
The intuition is that, in the second case, the first few liquidity buckets dominate more in the
convolution calculation and the central limit effect is mitigated.

Considering the different generalized hyperbolic special cases we see that the ratios are
usually largest for the t distribution followed by the other three
distributions; the exact ordering depends on the confidence level
$\alpha$ used in the calculation. In other words, use of the
Basel liquidity formula is least conservative in the case of t and more
conservative for the other distributions.

When we look at the confidence level of $\alpha=0.975$  which is the
level used in the new capital standard~\citep{bib:basel-16} the normal
inverse Gaussian (NIG) distribution leads to the highest level of
conservatism. This distribution is often a plausible model in market
risk applications. The ratio in the case where $n=5$ and $\rho=0$ is 0.837
which means that the Basel liquidity formula would tend to overstate
capital by around 19.4\%.

\begin{table}[htbp]
  \centering
  \begin{tabular}{*{2}{l}*{6}{r}}
    \toprule
     & \( \alpha \) & \multicolumn{2}{c}{0.95} & \multicolumn{2}{c}{0.975} & \multicolumn{2}{c}{0.99} \\
    \cmidrule(lr){3-4} \cmidrule(lr){5-6} \cmidrule(lr){7-8}
    Model & Quantity \textbar\ \( \rho \) & \multicolumn{1}{c}{0} & \multicolumn{1}{c}{0.5} & \multicolumn{1}{c}{0} & \multicolumn{1}{c}{0.5} & \multicolumn{1}{c}{0} & \multicolumn{1}{c}{0.5} \\
    \midrule
    Gauss & $c_{\alpha,\psi_1}$ & 2.063 & 2.063 & 2.338 & 2.338 & 2.665 & 2.665 \\
    & $c_{\alpha,\psi_L}$ & 2.063 & 2.063 & 2.338 & 2.338 & 2.665 & 2.665 \\
    & $r_\alpha$ & 1.000 & 1.000 & 1.000 & 1.000 & 1.000 & 1.000 \\ \addlinespace[3pt]
    t & $c_{\alpha,\psi_1}$ & 2.223 & 2.223 & 2.906 & 2.906 & 4.065 & 4.065 \\
    & $c_{\alpha,\psi_L}$ & 2.212 & 2.169 & 2.831 & 2.671 & 3.868 & 3.486 \\
    & $r_\alpha$ & 0.995 & 0.975 & 0.974 & 0.919 & 0.952 & 0.858 \\ \addlinespace[3pt]
    VG & $c_{\alpha,\psi_1}$ & 2.345 & 2.345 & 2.841 & 2.841 & 3.509 & 3.509 \\
    & $c_{\alpha,\psi_L}$ & 2.247 & 2.132 & 2.670 & 2.468 & 3.225 & 2.891 \\
    & $r_\alpha$ & 0.958 & 0.909 & 0.940 & 0.869 & 0.919 & 0.824 \\ \addlinespace[3pt]
    Hyp & $c_{\alpha,\psi_1}$ & 2.330 & 2.330 & 2.816 & 2.816 & 3.459 & 3.459 \\
    & $c_{\alpha,\psi_L}$ & 2.237 & 2.128 & 2.653 & 2.459 & 3.194 & 2.877 \\
    & $r_\alpha$ & 0.960 & 0.913 & 0.942 & 0.873 & 0.923 & 0.832 \\ \addlinespace[3pt]
    NIG & $c_{\alpha,\psi_1}$ & 2.374 & 2.374 & 2.976 & 2.976 & 3.832 & 3.832 \\
    & $c_{\alpha,\psi_L}$ & 2.296 & 2.167 & 2.801 & 2.544 & 3.502 & 3.042 \\
    & $r_\alpha$ & 0.967 & 0.913 & 0.941 & 0.855 & 0.914 & 0.794 \\
    \bottomrule
  \end{tabular}
\caption{Constants $c_{\alpha,\psi_{1}}$, $c_{\alpha,\psi_L}$ and
  ratios $r_\alpha$ in the experiment with 2 risk factors.}
\label{tab: exp1 bck=2}
\end{table}

\begin{table}[htbp]
  \centering
  \begin{tabular}{*{2}{l}*{6}{r}}
    \toprule
     & \( \alpha \) & \multicolumn{2}{c}{0.95} & \multicolumn{2}{c}{0.975} & \multicolumn{2}{c}{0.99} \\
    \cmidrule(lr){3-4} \cmidrule(lr){5-6} \cmidrule(lr){7-8}
    Model & Quantity \textbar\ \( \rho \) & \multicolumn{1}{c}{0} & \multicolumn{1}{c}{0.5} & \multicolumn{1}{c}{0} & \multicolumn{1}{c}{0.5} & \multicolumn{1}{c}{0} & \multicolumn{1}{c}{0.5} \\
    \midrule
    Gauss & $c_{\alpha,\psi_1}$ & 2.063 & 2.063 & 2.338 & 2.338 & 2.665 & 2.665 \\
    & $c_{\alpha,\psi_L}$ & 2.063 & 2.063 & 2.338 & 2.338 & 2.665 & 2.665 \\
    & $r_\alpha$ & 1.000 & 1.000 & 1.000 & 1.000 & 1.000 & 1.000 \\ \addlinespace[3pt]
    t & $c_{\alpha,\psi_1}$ & 2.223 & 2.223 & 2.906 & 2.906 & 4.065 & 4.065 \\
    & $c_{\alpha,\psi_L}$ & 2.160 & 2.169 & 2.637 & 2.671 & 3.402 & 3.486 \\
    & $r_\alpha$ & 0.972 & 0.975 & 0.908 & 0.919 & 0.837 & 0.858 \\ \addlinespace[3pt]
    VG & $c_{\alpha,\psi_1}$ & 2.345 & 2.345 & 2.841 & 2.841 & 3.509 & 3.509 \\
    & $c_{\alpha,\psi_L}$ & 2.112 & 2.132 & 2.429 & 2.468 & 2.824 & 2.891 \\
    & $r_\alpha$ & 0.901 & 0.909 & 0.855 & 0.869 & 0.805 & 0.824 \\ \addlinespace[3pt]
    Hyp & $c_{\alpha,\psi_1}$ & 2.330 & 2.330 & 2.816 & 2.816 & 3.459 & 3.459 \\
    & $c_{\alpha,\psi_L}$ & 2.108 & 2.128 & 2.423 & 2.459 & 2.814 & 2.877 \\
    & $r_\alpha$ & 0.905 & 0.913 & 0.860 & 0.873 & 0.813 & 0.832 \\ \addlinespace[3pt]
    NIG & $c_{\alpha,\psi_1}$ & 2.374 & 2.374 & 2.976 & 2.976 & 3.832 & 3.832 \\
    & $c_{\alpha,\psi_L}$ & 2.142 & 2.167 & 2.492 & 2.544 & 2.942 & 3.042 \\
    & $r_\alpha$ & 0.902 & 0.913 & 0.837 & 0.855 & 0.768 & 0.794 \\
    \bottomrule
  \end{tabular}
\caption{Constants $c_{\alpha,\psi_{1}}$, $c_{\alpha,\psi_L}$ and
  ratios $r_\alpha$ in the experiment with 5 risk factors.}
\label{tab: exp1 bck=5}
\end{table}

\section{Conclusion}\label{sec:conclusion}

We have presented evidence that the Basel liquidity formula tends to
lead to conservative capital charges when financial risk factors come
from heavier-tailed elliptical distributions.

The Basel formula is clearly a heavily stylized formula and makes a
number of crude assumptions. We have concentrated on the effect
of changing the underlying distribution of the risk factors when
portfolio sensitivities are linear. However, there are other important effects
we have not considered which will have an influence on the ability of
the formula to capture risk. In particular, the true effect of
risk-factor changes on portfolio risk is likely to be highly non-linear
over the kind of time horizons we consider. Moreover, as we have
already noted, positive serial correlation between losses over
different sub-intervals $[h_{k-1},h_k]$ of the overall liquidity
horizon $[0,h_n]$ will tend to lead to a tendency towards
underestimation which may counteract the central limit effect.

It should also be noted that there are many further layers of conservatism 
built into the new system of risk charges for the trading book, such
as the requirement to calibrate the model to stress periods and the
requirement to adjust the calculation to understate the possible diversification
effects across risk factors.

Nonetheless it is important to be clear about the workings of the
formula and the extent to which it may be interpreted as a
principles-based approach to the measurement of market risk. Our study
should be understood as a contribution to the clarification of this issue.


  \renewcommand{\theequation}{A.\arabic{equation}}
  \setcounter{equation}{0}  

\section*{Appendix}\label{sec:char-gener-some}

A standardized univariate generalized hyperbolic random variable $Y$
has the
stochastic representation $Y= \sqrt{W}V$ where $V$ is a standard normal
variable and $W$ is an independent positive random variable with a
generalized-inverse-Gaussian (GIG) distribution. The density of the
latter is
\begin{equation}\label{eq:11}
f_W(w)=\frac{\chi^{-\lambda}(\sqrt{\chi\kappa})^{\lambda}}{2K_{\lambda}(\sqrt{\chi\kappa})}w^{\lambda-1}
\exp(-\tfrac{1}{2}(\chi w^{-1}+\kappa w)),\quad
\begin{cases}\chi >0,\kappa\geq0&\text{if $\lambda < 0$}\\
\chi>0, \kappa>0 &\text{if $\lambda=0$}\\
\chi\geq 0, \kappa>0 &\text{if $\lambda>0$}
\end{cases}
\end{equation}
where $K_\lambda$ denotes a Bessel function of the third kind. The characteristic function of $Y$ is given by
\begin{equation}\label{eq:19}
  \phi_Y(s) = \E\left(\E\left(\exp(is\sqrt{W}V) \mid W\right)\right) =
  \E\left(\exp(-\tfrac{1}{2}s^2W)\right) = \int_{0}^\infty
  e^{-\tfrac{1}{2}s^2 w}f_W(w) \rd w
\end{equation}
and the variance by $\var(Y) = \E(W)$.

We first consider the case where $\chi>0$ and $\kappa>0$. In this case the variance of $Y$ is 
\begin{equation}\label{eq:18}
\var(Y) = \left( \frac{\chi}{\kappa} \right)^{1/2} \frac{K_{\lambda+1}(\sqrt{\chi\kappa})}{K_{\lambda}(\sqrt{\chi\kappa})}
\end{equation}
and the characteristic function is
\begin{align}
  \phi_Y(s) &=\int_{0}^{\infty}e^{-\tfrac{1}{2}\left(\chi
              w^{-1}+(s^2+\kappa)w\right)}\frac{\chi^{-\lambda}\left(\chi\kappa\right)^{\lambda/2}}{2K_{\lambda}\left(\sqrt{\chi\kappa}\right)}x^{\lambda-1}{\rm
              d}w \nonumber \\
 &=\left(\frac{\kappa}{s^2+\kappa}\right)^{\lambda/2}\frac{K_{\lambda}\left(\sqrt{\chi(s^2+\kappa)}\right)}{K_{\lambda}\left(\sqrt{\chi\kappa}\right)}. \label{eq:14}
\end{align}
 

We next consider the case of a Student t distribution which
corresponds to  $\kappa=0$, $\lambda = -\nu/2$ and $\chi=\nu$. In this
case $W$ has an
  inverse gamma distribution $W \sim \text{IG}(\nu/2,\nu/2)$ and
  $\var(Y)=\E(W) = \nu/(\nu-2)$, provided $\nu>2$. The characteristic
 function should be interpreted as the limit
  of~\eqref{eq:14} as $\kappa\to 0$.
Substituting the density of an inverse gamma distribution
into~\eqref{eq:19} yields
\begin{align}
  \phi_Y(s) &= 
 \int_0^\infty e^{-\tfrac{1}{2} s^2 w}
   \frac{(\frac{1}{2}\nu)^{\nu/2}}{\Gamma(\frac{1}{2}\nu)}
   w^{-\frac{\nu}{2}-1}e^{-\frac{1}{2} \nu w^{-1}} {\rm d} w
              \nonumber \\
  &= \frac{(\nu s^2)^{\nu/4}}{2^{\nu/2-1}\Gamma(\frac{1}{2}\nu)}K_{\nu/2}(\sqrt{\nu s^2}). \label{eq:13}
\end{align}

The special case of variance gamma (VG) corresponds to 
  $\chi=0$; without loss of generality we set the scaling parameter $\kappa=2$. In this case $W$ has a
  gamma distribution $W \sim \text{Ga}(\lambda,1)$ and
  $\var(Y)=\E(W)=\lambda$. The characteristic
  function in this case should be interpreted as the limit
  of~\eqref{eq:14} as $\chi\to 0$. Substituting the density of a gamma distribution
$W\sim\text{Ga}(\lambda,1)$ for $f_W$
in~\eqref{eq:19} we obtain
\begin{align}
\phi_Y(s) &
          =\int_{0}^{\infty}e^{-\frac{1}{2}s^{2}w}\frac{w^{\lambda-1}e^{-
          w}}{\Gamma(\lambda)}{\rm d}w \nonumber \\
&= \left(1+\tfrac{1}{2}s^2\right)^{-\lambda}.   \label{eq:167}
\end{align}

Two further special cases are the normal inverse Gaussian (NIG) and
hyperbolic distributions. In both cases we fix the parameter $\lambda$
and reparameterize the GH distribution in terms of
$\alphanew=\sqrt{\chi\kappa}$ and $\kappa$; the latter then appears
only
 as a scaling parameter and can be set to one. 

For the NIG
distribution $\lambda=-1/2$ and $\var(Y)=\alphanew^{-1}$. The identity 
  $K_\lambda(x)=K_{-\lambda}(x)$ can be used to infer that
  \begin{equation}
    \label{eq:17}
    \phi_Y(s) = \left(
\frac{\sqrt{\alphanew^2 + s^2}}{\alphanew}
\right)^{1/2}
    \frac{K_{1/2}\left(\sqrt{\alphanew^2 +
          s^2}\right)}{K_{1/2}\left(\alphanew \right)}.
  \end{equation}
For the hyperbolic (Hyp) distribution $\lambda=1$ and
$\var(Y)=\alphanew^{-1}K_2(\alphanew)/K_1(\alphanew)$. The
characteristic function is
\begin{equation}
    \label{eq:177}
    \phi_Y(s) = \left(
\frac {\alphanew}{\sqrt{\alphanew^2 + s^2}}
\right)
    \frac{K_{1}\left(\sqrt{\alphanew^2 +
          s^2}\right)}{K_{1}\left(\alphanew \right)}.
  \end{equation}

 \bibliographystyle{jf}  

\newcommand{\noopsort}[1]{}

\end{document}